\documentclass[a4paper, 12pt]{article}

\usepackage[sort&compress]{natbib}
\bibpunct{(}{)}{;}{a}{}{,} 

\usepackage{amsthm, amsmath, amssymb, mathrsfs, multirow, url, subfigure}
\usepackage{graphicx} 
\usepackage{ifthen} 
\usepackage{amsfonts}
\usepackage[usenames]{color}
\usepackage{fullpage}

\theoremstyle{plain} 
\newtheorem{theorem}{Theorem}

\theoremstyle{definition}

\theoremstyle{remark}

\newcommand{\prob}{\mathsf{P}} 
\newcommand{\E}{\mathsf{E}}
\newcommand{\V}{\mathsf{V}}

\newcommand{\nm}{{\sf N}}

\newcommand{\RR}{\mathbb{R}}

\renewcommand{\phi}{\varphi} 
\newcommand{\eps}{\varepsilon}

\newcommand{\Jqv}{\langle J \rangle} 
\newcommand{\Jqvhat}{\widehat{\Jqv}_n}
\newcommand{\Jind}{\{J\}}

\title{Efficient posterior inference on the volatility of a jump diffusion process}
\author{Ryan Martin \\
Department of Statistics \\
North Carolina State University \\
{\tt rgmarti3@ncsu.edu} \\
\mbox{} \\
Cheng Ouyang  \quad and \quad Francois Domagni \\
Department of Mathematics, Statistics, and Computer Science \\
University of Illinois at Chicago \\
{\tt (couyang, fdomag2)@uic.edu} 
}
\date{\today}

\begin{document}

\maketitle

\begin{abstract}
Jump diffusion processes are widely used to model asset prices over time, mainly for their ability to capture complex discontinuous behavior, but inference on the model parameters remains a challenge.  Here our goal is posterior inference on the volatility coefficient of the diffusion part of the process based on discrete samples.  A Bayesian approach requires specification of a model for the jump part of the process, prior distributions for the corresponding parameters, and computation of the joint posterior.  Since the volatility coefficient is our only interest, it would be desirable to avoid the modeling and computational costs associated with the jump part of the process.  Towards this, we consider a {\em purposely misspecified model} that ignores the jump part entirely.  We work out precisely the asymptotic behavior of the Bayesian posterior under the misspecified model, propose some simple modifications to correct for the effects of misspecification, and demonstrate that our modified posterior inference on the volatility is efficient in the sense that its asymptotic variance equals the no-jumps model Cram\'er--Rao bound.  

\smallskip

\end{abstract}

\section{Introduction}
\label{S:intro}

Jump diffusion models have gained considerable attention in the last two decades, especially in finance and economics, where they are used to model asset prices as a function of time.  An advantage of these models, over the classical Black--Scholes models \citep[e.g.][]{musiela.rutkowski.2005}, based solely on a continuous Brownian motion, is their ability to accommodate the rapid---seemingly discontinuous---changes in asset prices often seen in applications.  In fact, several authors have concluded that neither a purely-continuous nor purely-jump model is sufficient for real applications \citep[e.g.,][]{sahalia.jacod.2009, sahalia.jacod.2010, bn.shephard.2006, podolskij2006}. More specifically, by comparing the observed behavior of at-the-money and out-of-the-money call option prices near expiration with their analogous theoretical behavior, \citet{carr.wu.2003} and \citet{medvedev.scaillet.2007} argued that both continuous and jump components are necessary to explain the implied volatility behavior of S\&P500 index options.  In this paper we consider a continuous-time process $X=(X_t: t \in [0,T])$ over a fixed and finite time horizon $[0,T]$ that can be decomposed as 
\begin{equation}
\label{eq:levy}
X_t = \beta t + \theta^{1/2} W_t + J_t, \quad t \in [0,T], 
\end{equation}
where $\beta t + \theta^{1/2} W_t$ is a continuous diffusion---with $\beta$ the drift coefficient, $\theta$ the volatility coefficient, and $(W_t: t \in [0,T])$ a standard Brownian motion---and $J=(J_t: t \in [0,T])$ is a pure jump process with finite activity and independent of $W$. We emphasize here that we only assume that, with probability~1, the jump process $J$ has a finite number of jumps in $[0,T]$ and that each jump is finite.  For example, the results herein cover the case where $J$ is a compound Poisson process.  The quantity of interest here is the volatility coefficient $\theta$, a fundamentally important measure of uncertainty or risk \citep{musiela.rutkowski.2005}.  Our goal is to construct a (Bayesian) posterior distribution or, more generally, a data-dependent measure, $\Pi_n$ on $\RR^+ := (0,\infty)$ that can be used to provide valid uncertainty quantification about the volatility coefficient. 

If the entire process $X$ were observable, then we could immediately identify the jumps and, by subtraction, this could be converted to a standard problem.  However, here, as is typically the case in practice, the process $X$ is not fully observable; in particular, we can only observe $X$ at $n$ fixed times $0 < t_1 < t_2 < \cdots < t_n < T$, like in, e.g., \citet{sahalia.jacod.2009} and \citet{figueroa.lopez.2009}.  Having only discrete-time observations means that the continuous and jump parts cannot be separated with certainty, which forces us to deal with the jump part of the process in some way, even though our interest is only in the volatility of the continuous part.   A Bayesian approach would proceed by modeling all unknowns, using Bayes's theorem to get a posterior distribution for all the unknowns, and then evaluating the marginal posterior for inference on $\theta$.   While this approach is straightforward in principle, there are a number of challenges faced in practice.  For models more complex than that in \eqref{eq:levy}, where the volatility depends on the sample path itself, sophisticated computational tools and/or approximation methods are needed to evaluate the likelihood function and sample from the full posterior \citep[e.g.,][]{goncalves.roberts.2014, johannes.polson.2009, beskos.papa.roberts.2006, sahalia2002, sahalia2006, golightly2009, casella.roberts.2011}.  Even in the relatively simple model \eqref{eq:levy}, being obligated to model the jump process $J$ has some undesirable consequences, especially since we are only interested in the volatility coefficient $\theta$.  Indeed, developing a sound model for $J$, and specifying reasonable priors for the corresponding model parameters, is a non-trivial task---how large and how frequent are the jumps?~is the jump size and rate constant in time?~etc---and the quality of marginal inference on $\theta$ depends critically on the quality of this posited model, which is unverifiable.  \citet{rifo.torres.2009}, for example, in a setting similar to ours in \eqref{eq:levy}, propose a Bayesian model that assumes $J$ is a Poisson process, which certainly would not be appropriate for all applications.  To avoid potential bias from model misspecification, one could go semi-parametric, e.g., characterize $J$ by its L\'evy measure/density and put a prior on that, but this severely complicates the posterior computation and, furthermore, the addition of an infinite-dimensional nuisance parameter may affect the efficiency of the marginal inference on $\theta$.  Frequentist approaches \citep[e.g.,][]{sahalia.jacod.book} are available to estimate $\theta$ without specifying a model for $J$, so it would be desirable to have a Bayesian(-like) counterpart that provides a full posterior distribution for uncertainty quantification.  


Towards this Bayesian counterpart, we consider, in Section~\ref{S:model}, a purposely misspecified model that completely ignores the jumps, basically treating the observations as if they arise from a simple diffusion model.  This misspecified model is highly regular and computationally convenient, so if not heavily influenced by misspecification, then perhaps it would suffice for valid inference on $\theta$.  A special case of our Theorem~\ref{thm:bvm} says that the misspecified posterior is asymptotically normal but the misspecification has some undesirable effects, namely, the center is off-target and the spread is too large.  Rather than abandon the misspecified model, we propose, in Section~\ref{S:correcting}, to correct for the effects of misspecification, by making two simple adjustments: a suitable scaling of the log-likelihood to correct the spread, and a location shift.  Both of these adjustments rely on us having a suitable estimator of the quadratic variation of the jump process $J$.  We then show, in Theorem~\ref{thm:fixed}, that the corresponding modified posterior is asymptotically normal, centered around a consistent estimator of the true volatility, with variance equal to the Cram\'er--Rao lower bound for optimal/ideal case when there are no jumps, i.e., when the misspecified model is correct.  As a consequence, no proper Bayesian approach---with a parametric or nonparametric model for $J$---can do better asymptotically than our proposal.   Our particular modification is easy to implement and we present some simulation results in Section~\ref{S:numerical} to illustrate the validity of our modified posterior credible intervals.    

``Misspecification on purpose'' is a general idea which is both practically useful and theoretically interesting, with applications beyond the constant-volatility, jump diffusion setup considered here.  Our choice to demonstrate the benefits of this general idea in a relatively simple setting is only for the sake of clarity and conciseness.  Similar analysis applies in more complex situations but, naturally, the details (work in progress) are more involved and would potentially distract from the general idea.



\section{A misspecified model}
\label{S:model}

Assume that we observe the continuous-time process $(X_t)$ at $n$ distinct time points $t_1 < \cdots < t_n$, i.e., our observations are $X_{t_1}, \ldots, X_{t_n}$; for notational convenience later on, set $t_0 = 0$ and $X_0 \equiv 0$.  For notational simplicity, we will assume that the time points are equally spaced, so that each time difference $t_i - t_{i-1}$ equals $\Delta_n = T n^{-1}$; the case of non-equally spaced sampling can be handled similarly.  To avoid dealing directly with the jump component of the model \eqref{eq:levy}, we consider a purposely misspecified model that ignores both the drift\footnote{Ignoring the drift part here is only for simplicity; if the drift is also of interest, it is straightforward to carry out the subsequent analysis on a joint $(\beta,\theta)$ posterior.} and the jump part, i.e., it assumes that the differences $D_i = X_{t_i}-X_{t_{i-1}}$, $i=1,\ldots,n$, are iid $\nm(0, \theta \Delta_n)$ for some $\theta > 0$.  This misspecified model is easy to work with and has no nuisance parameters so, if it---or a simple modification thereof---also provides valid inference on the volatility, then it ought to be useful.  The likelihood function for this misspecified model, up to proportionality constants, is given by 
\begin{equation}
\label{eq:lik}
L_n(\theta) = \theta^{-n/2} \exp\Bigl\{-\frac{1}{2\Delta_n \theta} \sum_{i=1}^n D_i^2 \Bigr\} = \theta^{-n/2} \exp\Bigl\{ -\frac{n}{2} \frac{\hat\theta_n}{\theta} \Bigr\}, 
\end{equation}
where 
\[ \hat\theta_n = (n\Delta_n)^{-1} \sum_{i=1}^n D_i^2 = T^{-1} \sum_{i=1}^n D_i^2, \]
is the maximum likelihood estimator.  Just like in the familiar Bayes approach, we introduce a prior distribution $\Pi$ for $\theta$, with density function $\pi$.  Here we consider a generalization of the Bayesian setup, defining the (pseudo-)posterior distribution as
\begin{equation}
\label{eq:post}
\Pi_n(A) = \frac{\int_A L_n(\theta)^{1/\kappa_n} \pi(\theta) \,d\theta}{\int L_n(\theta)^{1/\kappa_n}\pi(\theta) \,d\theta}, \quad A \subseteq \RR^+, 
\end{equation}
where $\kappa_n$ is a suitable (possibly stochastic) sequence to be specified.  The distribution $\Pi_n$ in \eqref{eq:post} is sometimes referred to as a ``Gibbs posterior'' \citep[e.g.,][]{zhang2006a, zhang2006b, jiang.tanner.2008, grunwald.ommen.2014, syring.martin.scaling, bissiri.holmes.walker.2016} and $\kappa_n$ is a ``temperature'' parameter; the case $\kappa_n \equiv 1$ corresponds to the usual Bayes posterior.  Unlike in the well-specified Bayesian setting, where posterior consistency is typical, our model being misspecified means that we cannot expect $\Pi_n$ to converge to a point mass at the true volatility coefficient.  Therefore, some correction will be needed to point our posterior towards the true volatility coefficient, but first we need to understand how $\Pi_n$ in \eqref{eq:post} behaves without any intervention on our part.  

A fundamental result in Bayesian asymptotics is the Bernstein--von Mises theorem, which states that, under certain regularity conditions, a suitably centered and scaled version of the posterior will resemble a normal distribution, in the sense that the total variation distance between that normalized posterior and the normal distribution converges to zero in probability.  This classical version is typically used in the case of a well-specified model, but recently there has been work on a version of the Bernstein--von Mises theorem for misspecified models.  In particular, \citet{kleijn.vaart.2012}, in their Theorem~2.1, give a Bernstein--von Mises theorem when the model is misspecified.  Our result that follows is based on their approach. 

Before stating the result, we need to introduce some notation.  Let $\prob^\star$ denote the distribution of the differences $(D_1,\ldots,D_n)$, with $D_i = X_{t_i} - X_{t_{i-1}}$, under the jump diffusion model, and $\prob_J^\star$ the corresponding conditional distribution, given the jump part $J$ of the process \eqref{eq:levy}.  Also, let $\beta^\star$ and $\theta^\star$ denote the true drift and volatility coefficients, and define the expectation, conditional expectation, variance, and conditional variance as $\E^\star$, $\E_J^\star$, $\V^\star$, and $\V_J^\star$, respectively.  We consider a ``high-frequency'' scenario \citep[e.g.,][]{sahalia.jacod.book}, so $n$ is large and it is safe to assume that, {{with probability~1}}, the time windows $[t_{i-1}, t_i)$ contain at most one jump.  Therefore, for almost all $J$, under $\prob_J^\star$, we have that $(D_1,\ldots,D_n)$ are independent, $D_i \sim \nm(\beta^\star \Delta_n + \mu_i, \, \theta^\star \Delta_n)$, where 
\begin{equation}
\label{eq:mean}
\mu_i = J_{t_i} - J_{t_{i-1}}, \quad i=1,\ldots,n.
\end{equation}
For a given $J$, let $\Jqv = \sum_{i=1}^n \mu_i^2$ denote the quadratic variation of the jump process $J$, and let $\Jind$ denote the set of indices $i$ such that the window $[t_{i-1}, t_i)$ contains a jump, i.e., $\mu_i \neq 0$ if and only if $i \in \Jind$.  We assume that the process \eqref{eq:levy} has finite jump activity, so $|\Jind| \vee \Jqv < \infty$ with $\prob^\star$-probability~1. We also assume that $\kappa_n$ is a stochastic sequence and that there exists $\kappa^\dagger$, possibly depending on $J$, such that $\kappa_n \to \kappa^\dagger$ in $\prob_J^\star$-probability for the given $J$.  Finally, the point around which the posterior will concentrate is 
\[ \theta^\dagger = \theta^\star + T^{-1} \Jqv. \]
Note that both $\theta^\dagger$ and $\kappa^\dagger$ are constants with respect to the conditional distribution $\prob_J^\star$.  

\begin{theorem}
\label{thm:bvm}
Consider the pseudo-Bayesian posterior $\Pi_n$ in \eqref{eq:post} based on a prior $\Pi$ and the misspecified model with likelihood in \eqref{eq:lik}.  If the prior density $\pi$ is continuous and positive in a neighborhood of $\theta^\dagger$, then, for $\prob^\star$-almost all $J$, the posterior $\Pi_n$
is asymptotically normal in the sense that 
\[ d\bigl(\Pi_n, \, \nm(\hat\theta_n, 2\kappa^\dagger \theta^{\dagger 2} n^{-1}) \bigr) \to 0 \quad \text{in $\prob_J^\star$-probability as $n \to \infty$}, \]
where $d$ is the total variation distance.  The above conclusion also holds unconditionally, i.e., the above convergence is also in $\prob^\star$-probability.  
\end{theorem}

\begin{proof}
See the Appendix.  
\end{proof}

The theorem asserts that, for the ``high-frequency'' setting where $n$ is large, if the data-generating process \eqref{eq:levy} has finite jump activity, then the posterior will resemble a normal distribution centered around $\hat\theta_n$.  Since $\hat\theta_n$ converges to $\theta^\dagger$ (see the proof of the theorem), it follows that the posterior will resemble a normal distribution centered at $\theta^\dagger$.  This is different from the usual Bernstein--von Mises theorems found in the Bayesian literature in that the point around which the posterior concentrates depends on both parameters and a hidden portion of the data, namely, $\Jqv$.  

There are two seemingly undesirable consequences of misspecification.  The first, as alluded to above, is that $\Pi_n$ is biased in the sense that the point around which $\Pi_n$ concentrates is $\theta^\dagger$ instead of the true volatility coefficient $\theta^\star$.  The second is more subtle and concerns the spread of $\Pi_n$.  \citet[][Sec.~1]{kleijn.vaart.2012} point out that the asymptotic variance in their Bernstein--von Mises theorem may not agree with that for $\hat\theta_n$ based on M-estimation theory \citep[e.g.,][Ch.~5]{vaart1998}.  Indeed, the maximum likelihood estimator $\hat\theta_n$ for the misspecified model can be viewed as an M-estimator and will, therefore, be asymptotically normal, with asymptotic variance given by the so-called ``sandwich formula'' which, in this case, gives 
\[ \V_J^\star(\hat\theta_n) = \frac{2\theta^{\dagger 2}}{n} \Bigl\{ 1 - \Bigl(\frac{\Jqv}{T \theta^{\dagger}} \Bigr)^2 \Bigr\} + O(n^{-2}). \]
This follows from the calculations leading up to \eqref{eq:mse.formula} in the Appendix.  Up to order $n^{-1}$,  this closely resembles the asymptotic variance in Theorem~\ref{thm:bvm}; in particular, if $\kappa^\dagger$ were equal to the term in braces above, then the two variance formulas agree.  Note that the genuine Bayes posterior has $\kappa^\dagger = 1$ and, therefore, will have asymptotic variance larger than that in the above display.  Consequently, the Bayesian posterior credible intervals would be too large, making the inference inefficient.  Section~\ref{S:correcting} below describes how we can correct for these two undesirable consequences of misspecification.

\section{Correcting for misspecification}
\label{S:correcting}

As discussed above, there are two effects of the model misspecification on $\Pi_n$, both depending on the quadratic variation of the jump portion of the process.  To deal with these effects, we will need a suitable estimator of the quadratic variation $\Jqv$.  Intuitively, those observed differences $D_i$ which are of relatively large magnitude are likely due to jumps, so a reasonable estimator is 
\begin{equation}
\label{eq:Jhat}
\Jqvhat = \sum_{i=1}^n D_i^2 \, 1(|D_i| > \eta_n), 
\end{equation}
where $\eta_n$ is a sequence that vanishes sufficiently slow, and $1(\cdot)$ is the indicator function.  We claim that if $\eta_n \propto n^{-\omega}$ for some $\omega \in (0,\frac12)$, then 
\begin{equation}
\label{eq:bound2}
\E_J^\star|\Jqvhat - \Jqv| = O(n^{-1/2}), \quad n \to \infty, 
\end{equation}
for $\prob^\star$-almost all $J$.  Results of this type for L\'evy processes are available in the literature \citep[e.g.,][Fact~3.7]{sahalia.jacod.book}, but the proof of \eqref{eq:bound2} given in the Appendix under only the finite jump activity assumption is relatively simple.  With a suitable estimator in hand, now we are ready to address the effects of misspecification.

Towards constructing a modified posterior for the volatility, we must consider the following question: what is the ``correct/optimal'' asymptotic variance in the normal approximation?  Of course, the best possible inference obtains if the model is not misspecified, i.e., there are no jumps; this is equivalent to the case where the sample path of the process is fully observed since, in that case, the jumps are visible and can be removed.  An easy calculations reveals that, in this ideal case, the asymptotic variance is the Cram\'er--Rao bound, $2\theta^{\star 2} n^{-1}$.  This optimal variance obtains in Theorem~\ref{thm:bvm} if 
\[ \kappa^\dagger = \Bigl( \frac{\theta^\star}{\theta^\dagger} \Bigr)^2 = \Bigl( 1 - \frac{\Jqv}{T \theta^\dagger} \Bigr)^2. \]
This suggests we choose $\kappa_n$ in \eqref{eq:post} as 
\begin{equation}
\label{eq:kappa}
\kappa_n = \Bigl( 1 - \frac{\Jqvhat}{T \hat\theta_n} \Bigr)^2, 
\end{equation}
so that, by \eqref{eq:bound2}, $\kappa_n \to \kappa^\dagger$ in $\prob_J^\star$-probability for $\prob^\star$-almost all $J$.  With this understanding, we define a ``modified'' Bayesian posterior $\widetilde \Pi_n$ as the distribution of $\theta - T^{-1} \Jqvhat$ when $\theta$ is distributed as $\Pi_n$ in \eqref{eq:post}, with $\kappa_n$ as in \eqref{eq:kappa}.  In other words, if $\pi_n$ is the density function corresponding to $\Pi_n$, then $\widetilde \Pi_n$ has density function 
\begin{equation}
\label{eq:modified}
\tilde\pi_n(\theta) = \pi_n(\theta + T^{-1}\Jqvhat), \quad \theta \in \RR^+.
\end{equation}
Then we have the following Bernstein--von Mises theorem for $\widetilde \Pi_n$.   

\begin{theorem}
\label{thm:fixed}
Under the same setup as in Theorem~\ref{thm:bvm}, for $\prob^\star$-almost all $J$, the modified Bayesian posterior $\widetilde \Pi_n$, with $\kappa_n$ as in \eqref{eq:kappa}, satisfies  
\[ d\bigl(\widetilde\Pi_n, \, \nm(\hat\theta_n-T^{-1} \Jqvhat, 2\theta^{\star 2} n^{-1}) \bigr) \to 0 \quad \text{in $\prob_J^\star$-probability as $n \to \infty$}, \]
where $d$ is the total variation distance.  The above convergence is also in $\prob^\star$-probability.  
\end{theorem}

\begin{proof}
Since the total variation distance is invariant to location shifts, we have that 
\[ d\bigl(\widetilde\Pi_n, \, \nm(\hat\theta_n-T^{-1} \Jqvhat, 2\theta^{\star 2} n^{-1}) \bigr) = d\bigl(\Pi_n, \, \nm(\hat\theta_n, 2\theta^{\star 2} n^{-1}) \bigr). \]
That the right-hand side converges to 0 in $\prob_J^\star$-probability follows from Theorem~\ref{thm:bvm} and the particular choice of $\kappa_n$ in \eqref{eq:kappa}.  The $\prob^\star$ convergence is proved just like in Theorem~\ref{thm:bvm}.  
\end{proof}

The first observation is that, since $\hat\theta_n - T^{-1} \Jqvhat$ is a consistent estimator of $\theta^\star$, the modified posterior is concentrating around the true volatility coefficient, as desired.  Furthermore, by our choice of the sequence $\kappa_n$, the asymptotic variance agrees with that achieved in the ideal case where there are no jumps present or, equivalently, when the sample path of the process is fully observed.  The remaining question is if the posterior variance agrees with the variance of the center $\hat\theta_n - T^{-1} \Jqvhat$ under $\prob^\star$.  Proposition~1 in \citet{sahalia.jacod.2010} reveals that, in the present setting, under $\prob_J^\star$, the estimator $\hat\theta_n - T^{-1}\Jqvhat$ satisfies a central limit theorem, with asymptotic variance $2\theta^{\star 2} n^{-1}$.  Therefore, the credible intervals coming from the modified posterior $\widetilde \Pi_n$ will be asymptotically valid under $\prob_J^\star$, i.e., the coverage probability of the $100(1-\alpha)$\% credible intervals will converge to $1-\alpha$ for $\prob^\star$-almost all $J$.  It follows immediately from the dominated convergence theorem that the coverage probability converges to $1-\alpha$ under $\prob^\star$ as well.  It turns out that the finite-sample performance depends on the choice of threshold $\eta_n$ in \eqref{eq:Jhat} and we address this in Section~\ref{S:numerical}.

\section{Numerical results}
\label{S:numerical}

An important question is how to choose the threshold $\eta_n$.  The theory says that we need $\eta_n = m n^{-\omega}$ for some $m > 0$ and some $\omega \in (0, \frac12)$ but, in finite samples, $m$ and $\omega$ are not independent parameters; that is, only the value of $\eta_n$ matters, not the particular $(m,\omega)$.  This point is discussed at length in \citet[][Sec.~6.2.2]{sahalia.jacod.book}, and they suggest one reasonable strategy for choosing $\eta_n$.  We consider here a simpler approach based on outlier detection.  That is, let $Q$ denote the interquartile range of the observed increment magnitudes $|D_1|,\ldots,|D_n|$; this value is likely to be small since almost all of the increments correspond to the diffusion part of the process.  Take $\eta_n$ to be some value that lower-bounds the set of all $|D_i|$s that exceeds some cutoff, say, $5Q$.  We make no claims that this approach is ``optimal'' in any sense, only that it is both simple and reasonable.  

For illustration, consider the model \eqref{eq:levy} with drift $\beta^\star = 1$, volatility $\theta^\star = 10$, and compound Poisson process jumps with a rate of $\lambda=5$ jumps per unit time and jumps sampled from the discrete uniform distribution on $\{-\tau, \tau\}$ with $\tau = 3$.  We simulate $n=5000$ equally spaced observations from this process.  A plot of the observed sample path on the interval $[0,1]$ is shown in Figure~\ref{fig:example}(a) with the jump times highlighted by vertical lines.  For the misspecified Bayes model, we consider a conjugate inverse gamma prior with shape $a=1$ and rate $b=1$; the presence of the temperature parameter $\kappa_n$ does not affect conjugacy.  We also fix $\eta_n$ based on the interquartile range strategy described above.  Figure~\ref{fig:example}(b) shows the posterior density function \eqref{eq:modified}, the corresponding 95\% credible interval for the volatility, and the density function of the normal approximation in Theorem~\ref{thm:fixed}.  The key observations here are that modified posterior density is centered very close to the true volatility in this case, and hence the credible interval contains it, and also that the posterior density and the normal approximation are very similar.  Since the the normal approximation has variance equal to the Cram\'er--Rao lower bound under the benchmark no-jumps model, the plot Figure~\ref{fig:example}(b) reveals the overall efficiency of the modified posterior inference.     

\begin{figure}
\begin{center}
\subfigure[Sample path]{\scalebox{0.6}{\includegraphics{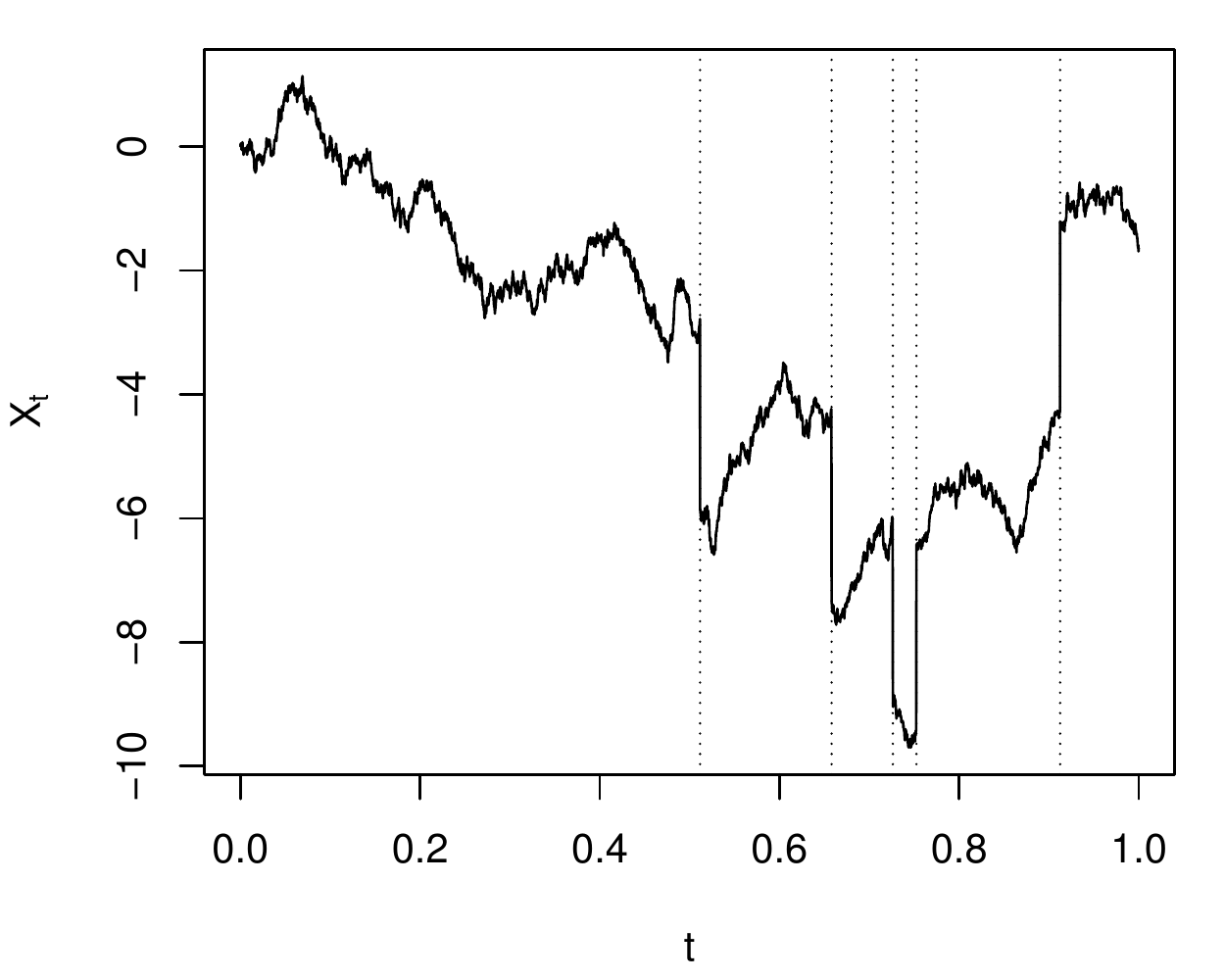}}}
\subfigure[Modified posterior density]{\scalebox{0.6}{\includegraphics{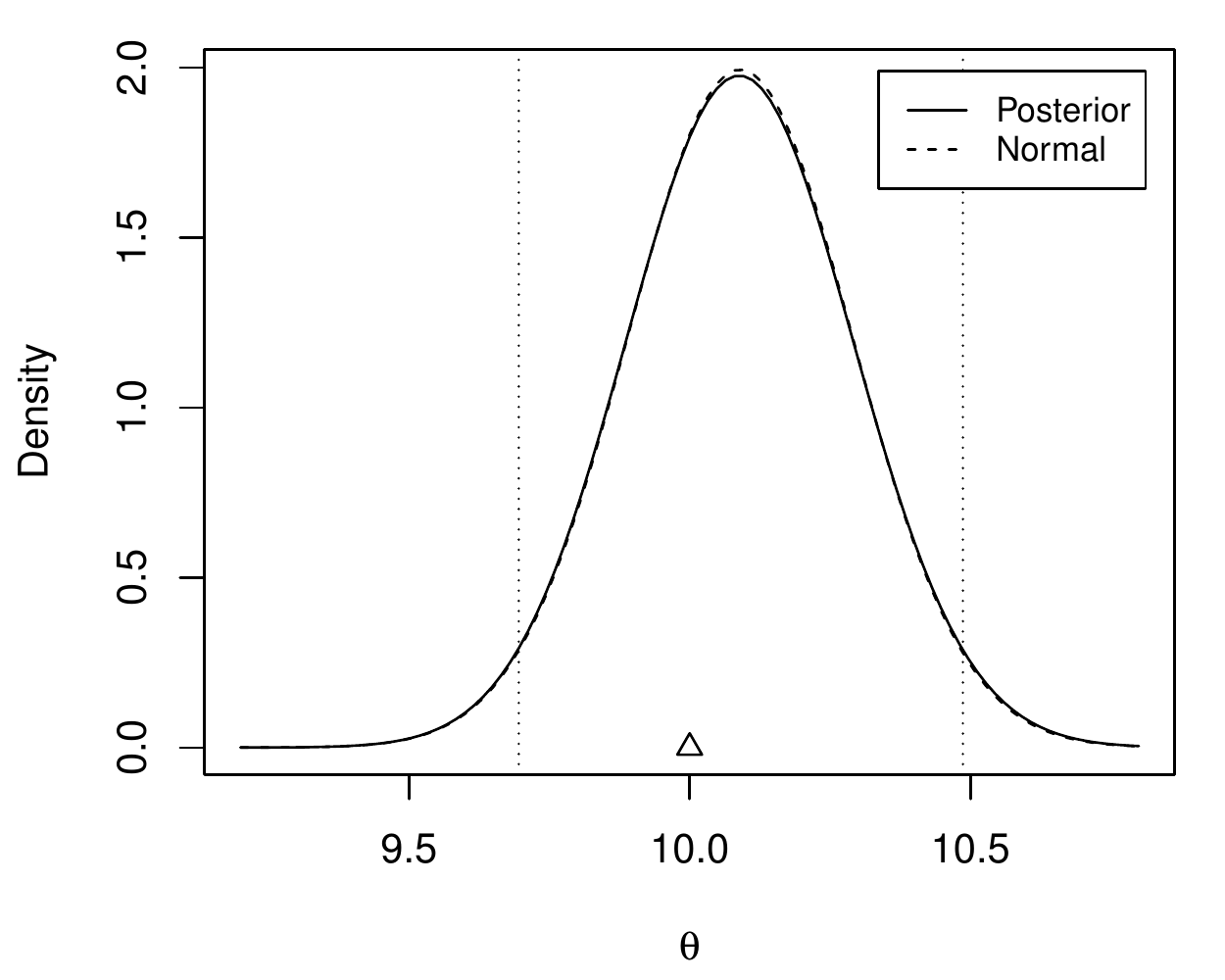}}}
\end{center}
\caption{Plot of one sample path $X$ (left) along with the corresponding density function for the modified posterior (right); also shown in the right panel is the 95\% modified posterior credible interval (vertical lines), the normal approximation in Theorem~\ref{thm:fixed}, and the true volatility coefficient ($\triangle$).}
\label{fig:example}
\end{figure}

To further investigate the finite-sample properties our proposed approach for inference on the volatility, we consider a simulation study.  Using the same model as above, but varying the jump rate $\lambda \in \{4, 8, 16, 32\}$, the jump magnitude $\tau \in \{1,2,4,8\}$, and the sample size $n$, we investigate the coverage probability of the modified pseudo-Bayesian credible intervals, based on the choice of threshold mentioned above.  Figure~\ref{fig:sim} displays the empirical coverage probability, based on 5000 Monte Carlo simulations in each setting, summarized over the jump rate and standard deviation, for several values of $n$.  This plot reveals that the choice of threshold based on the interquartile range performs reasonably well in this setting, giving coverage probabilities very close to the nominal 95\% level.

\begin{figure}[t]
\begin{center}
\scalebox{0.6}{\includegraphics{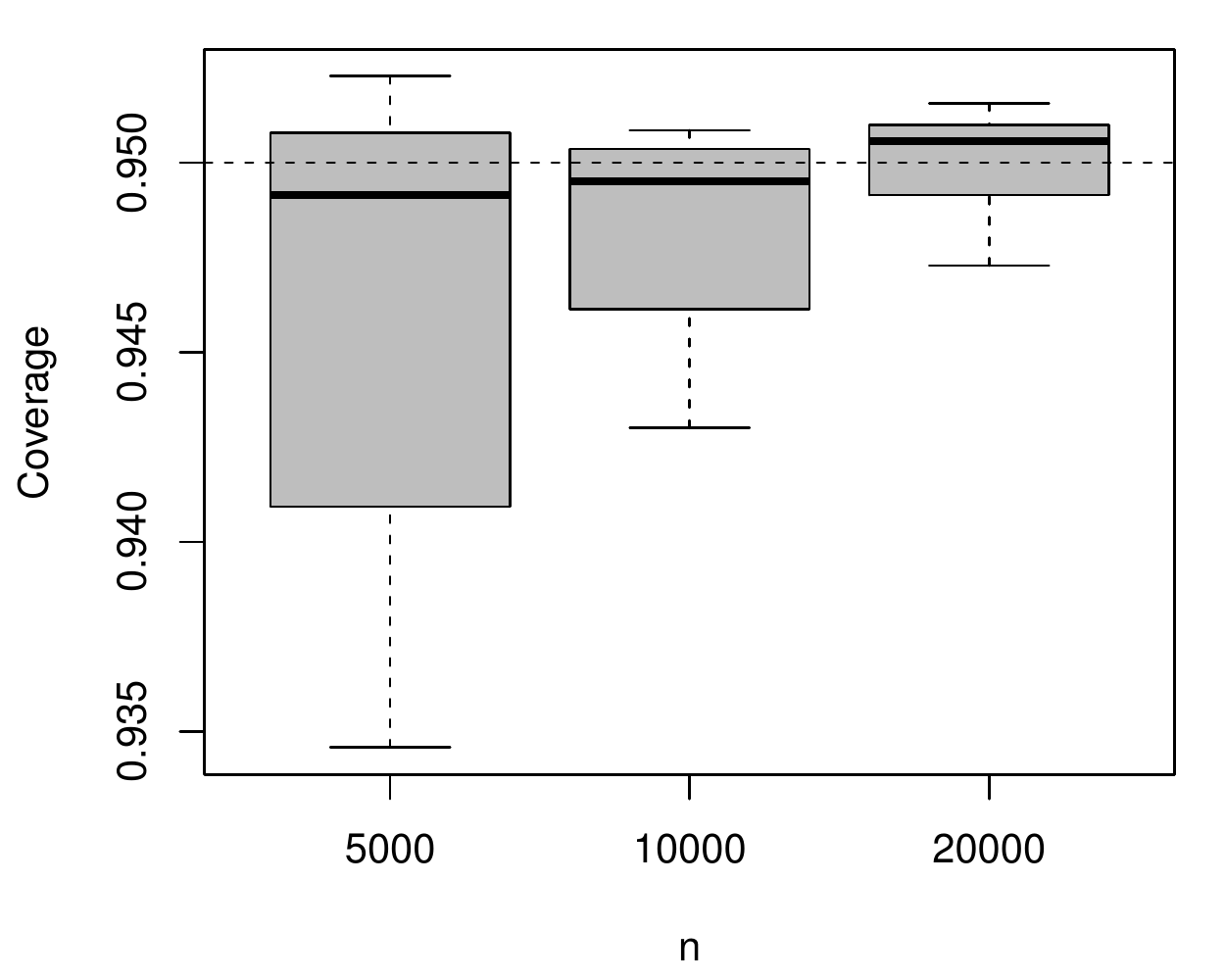}}
\end{center}
\caption{Summary of the empirical coverage probability of the 95\% modified posterior credible intervals, over various several values of the jump rate and standard deviation (see the text), and for several values of $n$.}
\label{fig:sim}
\end{figure}

\section{Concluding remarks}
\label{S:discuss}

In this paper we considered the construction of a (Bayesian-like) posterior for inference on the volatility coefficient in a jump diffusion model where the distribution for the jump part of the process is left unspecified.  By working with a ``purposely misspecified'' model we avoid the difficulties of modeling the jump part, as well as the corresponding computations, but the cost is some misspecification bias.  We correct for this bias by making two modifications: a non-trivial scaling of the likelihood, followed by a location shift.  We prove that this modified posterior is asymptotically normal and optimal in the sense that the asymptotic variance agrees with the Cram\'er--Rao lower bound under the no-jumps model.  Aside from the asymptotic results, these modifications are easy to implement and, as shown in Section~\ref{S:numerical}, gives valid inference in a range of examples.  


Of course, this ``misspecification on purpose'' strategy could be used in many other problems to provide valid uncertainty quantification for the parameters of interest without having to specify a complete model for the possibly very complex nuisance parameters, which is very attractive.  Furthermore, when one has reliable prior information about the interest parameter, it is straightforward to incorporate this into the proposed analysis compared to a fully non-Bayesian approach, say, using M-estimation with bootstrap.  

For this particular model, there are several extensions that one could consider.  For example, if the drift parameter was also of interest, then, instead of ignoring $\beta$ as we did here, it would be relatively straightforward to construct the same modified posterior for the pair $(\beta,\theta)$.  More interesting is the case where the volatility parameter is not a scalar constant but, instead, a function $\theta_t$.  Certain functionals of $\theta_t$, in particular, the average volatility $T^{-1} \int_0^T \theta_t \,dt$, can be inferred directly using virtually the same techniques as presented here.  Inference on the function $t \mapsto \theta_t$ itself would be more involved, but the analysis would be similar.  Application of the ``misspecification on purpose'' strategy in the case where the volatility is a function depending on the sample path itself, i.e., $\theta=\theta(X_t)$, is an interesting open problem.



\appendix

\section{Proofs}

\subsection{Proof of Theorem~\ref{thm:bvm}}

Without loss of generality, we will assume in the proof that $T=1$ and $\Delta_n = n^{-1}$.  To prove the Bernstein--von Mises theorem, follow the approach described in Theorem~2.1 of \citet{kleijn.vaart.2012}.  Our first objective is to show that 
\[ \E_J^\star \bigl[ \Pi_n(\{\theta: |\theta - \theta^\dagger| > M_n n^{-1/2}\}) \bigr] \to 0, \quad n \to \infty, \]
for any sequence of constants $M_n \to \infty$.  To establish this, we need only to study the posterior mean and variance.  That is, if $\E_{\Pi_n}$ denotes expectation with respect to the posterior $\Pi_n$, then, by Markov's inequality, we have 
\begin{equation}
\label{eq:post.bound}
\Pi_n(\{\theta: |\theta-\theta^\dagger| > M_n n^{-1/2}\}) \leq n M_n^{-2} \E_{\Pi_n}(\theta - \theta^\dagger)^2.
\end{equation}
To show that the expectation of the left-hand side in the above display vanishes, it suffices to show that  
\begin{equation}
\label{eq:moment.bounds}
\E_J^\star\{\E_{\Pi_n}(\theta - \theta^\dagger)^2\} = O(n^{-1}). 
\end{equation}
Towards this, we will use the Laplace approximation which says that, for suitable functions $g$, the posterior mean of $g(\theta)$ is 
\[ g(\hat\theta_n) \{1 + O(n^{-1})\}, \quad n \to \infty. \]
Therefore, in our case, if we apply the above to $g(\theta)=(\theta - \theta^\dagger)^2$, then we have 
\[ \E_{\Pi_n}(\theta - \theta^\dagger)^2 = (\hat\theta - \theta^\dagger)^2\{1 + O(n^{-1})\}. \]
Since the log-likelihood function for our misspecified model has a unique maximum $\hat\theta_n$ in the interior of the $\theta$-space and satisfies $(\log L_n)''(\hat\theta_n) < 0$, the big-oh term above is uniform in observations, i.e., the $O(n^{-1})$ term in the above display is a function of $(D_1,\ldots,D_n)$ that can be uniformly bounded by a constant times $n^{-1}$ and, in particular, the scaling by $\kappa_n^{-1}$ does not affect this conclusion.   Therefore, to get \eqref{eq:moment.bounds}, it suffices to show that
\begin{equation}
\label{eq:mse1}
\E_J^\star(\hat\theta_n - \theta^\dagger)^2 = O(n^{-1})
\end{equation}
Towards showing \eqref{eq:mse1}, we recall that $\hat\theta_n = \sum_{i=1}^n D_i^2$ where, under $\prob_J^\star$, $(D_1,\ldots,D_n)$ are independent with 
\[ D_i \sim \nm(\beta^\star \Delta_n + \mu_i, \theta^\star \Delta), \quad i=1,\ldots,n, \]
and $\mu_i$ are defined in \eqref{eq:mean}.  It follows that $\hat\theta_n$ has the same distribution as $\theta^\star \Delta_n$ times a non-central chi-square random variable, $Y$, with $n$ degrees of freedom and non-centrality parameter $\lambda = (\theta^\star \Delta_n)^{-1} \sum_{i=1}^n (\beta^\star \Delta_n + \mu_i)^2$.  In particular, 
\[ \E(Y) = n + \lambda \quad \text{and} \quad \V(Y) = 2(n + 2\lambda). \]
If we let $\V_J^\star$ denote variance with respect to $\prob_J^\star$, then we have that 
\begin{align*}
\E_J^\star(\hat\theta_n - \theta^\dagger)^2 & = \V_J^\star(\hat\theta_n) + \{\E_J^\star(\hat\theta_n) - \theta^\dagger\}^2 \\
& = (\theta^\star \Delta_n)^2 \V(Y) + \{(\theta^\star \Delta_n) \E(Y) - \theta^\dagger\}^2. 
\end{align*}
Plugging in the formulas for the mean and variance of $Y$ and simplifying, gives 
\begin{equation}
\label{eq:mse.formula}
\E_J^\star(\hat\theta_n - \theta^\dagger)^2 = 2\theta^\star \theta^\dagger n^{-1} + O(n^{-2}). 
\end{equation}
This is clearly $O(n^{-1})$, so we have established \eqref{eq:mse1}.  Note that this derivation depends on $J$ only through $|\Jind|$ and $\Jqv$.  Since \eqref{eq:mse1} implies \eqref{eq:moment.bounds}, we have proved the claimed posterior concentration rate result.  

Next, we need to demonstrate that the model is suitably regular.  More specifically, Theorem~2.1 in \citet{kleijn.vaart.2012} require that the model satisfies a certain local asymptotic normality property, i.e., the log-likelihood ratio has a quadratic approximation locally around the specified $\theta^\dagger$.  Since the misspecified model is so nice, it is a straightforward exercise to show that 
\[ \Bigl| \frac{1}{\kappa_n} \log \frac{L_n(\theta^\dagger + n^{-1/2} h)}{L_n(\theta^\dagger)} 
- V_{\theta^\dagger} \, Z_n(\theta^\dagger) \, h + \tfrac12 V_{\theta^\dagger} \, h^2 \Bigr| = o_{\prob_J^\star}(1), \]
where $V_{\theta^\dagger} = (2\kappa^\dagger \theta^{\dagger 2})^{-1}$ and $Z_n(\theta^\dagger) = n^{1/2}(\hat\theta_n - \theta^\dagger)$.  The above display holds uniformly on compact subsets of $h$, and it follows from \eqref{eq:mse1} that $Z_n(\theta^\dagger)$ is bounded in $\prob_J^\star$-probability.  Therefore, the assertion in Theorem~\ref{thm:bvm} follows from Kleijn and van der Vaart's.  

The extension of these results to the unconditional distribution, $\prob^\star$, is also straightforward.  Based on the finite jump activity assumption, all that we demonstrated above holds with $\prob^\star$-probability 1.  In particular, we have that, for any $\eps > 0$, 
\[ \prob_J^\star\bigl\{ d\bigl(\Pi_n, \, \nm(\hat\theta_n, 2\kappa^\dagger \theta^{\dagger 2} n^{-1}) > \eps \bigr\} \to 0, \quad \text{for $\prob^\star$-almost all $J$}. \]
Since this sequence is bounded and converges almost surely, it follows from the dominated convergence theorem that 
\[ \prob^\star\bigl\{ d\bigl(\Pi_n, \, \nm(\hat\theta_n, 2\kappa^\dagger \theta^{\dagger 2} n^{-1}) > \eps \bigr\} = \E^\star \bigl[ \prob_J^\star\bigl\{ d\bigl(\Pi_n, \, \nm(\hat\theta_n, 2\kappa^\dagger \theta^{\dagger 2} n^{-1}) > \eps \bigr\} \bigr] \to 0, \]
i.e., $d\bigl(\Pi_n, \, \nm(\hat\theta_n, 2\kappa^\dagger \theta^{\dagger 2} n^{-1}) \to 0$ in $\prob^\star$-probability.

\subsection{Proof of the claim in Equation \eqref{eq:bound2}}

For $\prob^\star$-almost all $J$, we have that $(D_1,\ldots,D_n)$ are independent under $\prob_J^\star$ with 
\[ D_i = \mu_i + Z_i \sim \nm(\beta^\star \Delta_n + \mu_i, \theta^\star \Delta_n), \quad i=1,\ldots,n \]
where $\mu_i$ are given in \eqref{eq:mean} and $\mu_i \neq 0$ only for those indices $i \in \Jind$.  To prove the claim in Equation \eqref{eq:bound2}, we split the indices to those that contain a jump (in $\Jind$) and those that do not (in $\Jind^c$).  Then we get 
\begin{align*}
\Jqvhat - \Jqv & = \sum_{i=1}^n D_i^2 \, 1(|D_i| > \eta_n) - \Jqv \\
& = \sum_{i \not\in \Jind} Z_i^2 \, 1(|Z_i| > \eta_n) + \sum_{i \in \Jind} Z_i^2 \, 1(|Z_i + J_i| > \eta_n) \\
& \qquad + 2 \sum_{i \in \Jind} \mu_i Z_i \, 1(|Z_i + \mu_i| > \eta_n) + \sum_{i \in \Jind} \mu_i^2 \, 1(|Z_i + \mu_i| \leq \eta_n).
\end{align*}
Take absolute value of both sides, apply the triangle inequality, and then take expectation.  This yields the inequality 
\begin{align*}
\E_J^\star|\Jqvhat - \Jqv| & \leq \sum_{i \not\in \Jind} \E_J^\star\{ Z_i^2 \, 1(|Z_i| > \eta_n) \} + \sum_{i \in \Jind} \E_J^\star\{Z_i^2 \, 1(|Z_i + \mu_i| > \eta_n)\} \\
& \qquad + 2 \sum_{i \in \Jind} \E_J^\star\{|\mu_i Z_i| \, 1(|Z_i + \mu_i| > \eta_n)\} + \sum_{i \in \Jind} \E_J^\star\{ \mu_i^2 \, 1(|Z_i + \mu_i| \leq \eta_n)\}.
\end{align*}
We will proceed by showing, one by one, that each of the four terms in the upper bound above is $O(n^{-1/2})$.  First, note that 
\[ Z_i \sim \nm(\beta^\star \Delta_n, \theta^\star \Delta_n), \quad i=1,\ldots,n, \] 
are iid and hence its fourth moment is bounded by a constant independent of $n$.  Then we have, by the Cauchy--Schwartz inequality
\begin{align*}
\sum_{i \not\in \Jind} \E_J^\star\{Z_i^2 \, 1(|Z_i| > \eta_n)\} & \lesssim (\E_J^\star|Z_1|^4)^{1/2} \sum_{i \not\in \Jind} (\prob(|Z_i| > \eta_n))^{1/2} \\
& \lesssim \,|\{J\}^c| \, \prob(|Z_i| > \eta_n)^{1/2}.
\end{align*}
It is clear that $|\Jind^c|$ is of order $n$.  So we only need to find a good bound for the tail probability.  Assume, for the moment, that $\beta^\star > 0$.  Using the usual normal tail probability bounds, we get 
\[ \prob(|Z_i| > \eta_n) \lesssim \frac{\Delta_n^{1/2}}{\eta_n-\beta^\star \Delta_n} e^{-(\eta_n - \beta^\star \Delta_n)^2 / \Delta_n} \lesssim \frac{\Delta_n^{1/2}}{\eta_n} e^{-\eta_n^2/\Delta_n}. \]
So, if $\eta_n \propto n^{-\omega}$, for $\omega \in (0,\frac12)$, then the upper bound for the above tail probability is $o(n^{-k})$ for any positive integer $k$. Hence it follows easily 
\[ \sum_{i \not\in \Jind} \E_J^\star\{Z_i^2 \, 1(|Z_i| > \eta_n)\} = o(n^{-1/2}); \]
the same conclusion can be reached if $\beta^\star < 0$.  Next, 
\[ \sum_{i \in \Jind} \E_J^\star\{Z_i^2 \, 1(|Z_i + \mu_i| > \eta_n)\} \leq |\Jind| \{\theta^\star \Delta_n + (\beta^\star \Delta_n)^2\} = O(n^{-1}) \]
and, similarly, using Cauchy--Schwartz, 
\[ \sum_{i \in \Jind} \E_J^\star\{|\mu_i Z_i| \, 1(|Z_i + \mu_i| > \eta_n)\} \lesssim \{\theta^\star \Delta_n + (\beta^\star \Delta_n)^2\}^{1/2} \sum_{i \in \Jind} |\mu_i| = O(n^{-1/2}). \]
For the last term, we need to bound $\prob_J^\star(|Z_i + \mu_i| \leq \eta_n)$.  Again, without loss of generality, if we assume that $\beta^\star > 0$ and $\mu_i > 0$, then we get
\[ \prob_J^\star(|Z_i + \mu_i| \leq \eta_n) \leq \prob\{ \nm(0,1) > \Delta_n^{-1/2}(\mu_i + \beta^\star \Delta_n - \eta_n) \}, \]
which, using the normal tail probability bound again, is bounded by 
\[ \frac{\Delta_n^{1/2}}{\mu_i + \beta^\star \Delta_n - \eta_n} e^{-(\mu_i + \beta^\star \Delta_n - \eta_n)^2/\Delta_n}. \]
Since $\mu_i > 0$ is a fixed constant, the above quantity vanishes exponentially fast, so 
\[ \sum_{i \in \Jind} \mu_i^2 \, \prob(|Z_i + \mu_i| \leq \eta_n) \leq o(n^{-1/2}) \sum_{i \in \Jind} \mu_i^2. \]
All four terms have been shown to be $O(n^{-1/2})$, completing the proof of \eqref{eq:bound2}.  Finally, note that the result holds for all $J$ such that $|\Jind|$ and $\Jqv = \sum_{i \in \Jind} \mu_i^2$ are finite.  Since this is a $\prob^\star$-probability~1 event, \eqref{eq:bound2} holds for $\prob^\star$-almost all $J$.

\ifthenelse{1=1}{}{
\subsection{Proof of Theorem~\ref{thm:fixed}}

For the first assertion, on the concentration rate for the modified posterior $\widetilde \Pi_n$, by Markov and triangle inequalities, we get 
\begin{align*}
\widetilde \Pi_n(\{\theta: |\theta-\theta^\star| > M_n n^{-1/2}\}) & = \Pi_n(\{\theta: |\theta - \theta^\star - \Jqvhat| > M_n n^{-1/2} \}) \\
& \leq n^{1/2} M_n^{-1} \, \E_{\Pi_n}|\theta - \theta^\star - \Jqvhat| \\
& = n^{1/2} M_n^{-} \, \E_{\Pi_n}|\theta - \theta^\star - \Jqv + \Jqv - \Jqvhat| \\
& \leq n^{1/2} M_n^{-1} \, \bigl\{ \E_{\Pi_n}|\theta-\theta^\star - \Jqv| + |\Jqvhat - \Jqv| \bigr\}.
\end{align*}
By Cauchy--Schwartz and \eqref{eq:mse1}, we have that 
\[ \E_J^\star\{\E_{\Pi_n}|\theta-\theta^\star - \Jqv|\} = O(n^{-1/2}). \]
Furthermore, the expectation of the second term is $O(n^{-1/2})$ by \eqref{eq:bound2}.  Putting everything together, we have 
\[ \E_J^\star\bigl[ \widetilde \Pi_n(\{\theta: |\theta-\theta^\star| > M_n n^{-1/2}\}) \bigr] \to 0, \]
which is exactly the first assertion of the theorem.

The second assertion follows directly from the result of Theorem~\ref{thm:bvm} and the fact that the total variation distance is invariant to location shifts.  That is, 
\[ d_{\text{\sc tv}}\bigl(\widetilde \Pi_n, \nm(\hat\theta - \Jqvhat, n^{-1} V_{\theta^\dagger}^{-1}) \bigr) = d_{\text{\sc tv}}\bigl(\Pi_n, \nm(\hat\theta, n^{-1} V_{\theta^\dagger}^{-1}) \bigr), \]
and the right-hand side converges in $\prob_J^\star$-probability by Theorem~\ref{thm:bvm}.  

{\color{red} Last part about extending to the unconditional distribution $\prob^\star$...}
}

\ifthenelse{1=0}{
\bibliographystyle{apalike}
\bibliography{/Users/rgmarti3/Dropbox/Research/mybib}
}{

}

\end{document}